\documentclass[conference, a4paper]{IEEEtran}
\usepackage{amssymb, amsmath, amsthm, cite}
\usepackage{array}
\usepackage{graphicx}
\usepackage{color}
\usepackage{bbm}
\usepackage{xcolor}

\newcommand{\beq}[1]{\begin{equation}\label{#1}}
\newcommand{\eeq}{\end{equation}}
\newcommand{\req}[1]{(\ref{#1})}
\newcommand{\bmu}[1]{\begin{multline}\label{#1}}
\newcommand{\emu}{\end{multline}}
\newcolumntype{P}[1]{>{\centering\arraybackslash}p{#1}}
\renewcommand{\(}{\left(}
\renewcommand{\)}{\right)}

\newcommand{\eq}{\triangleq}

\newcommand{\A}{\mathcal{A}}
\newcommand{\C}{\mathcal{C}}
\renewcommand{\P}{\mathcal{P}}

\newcommand{\E}{\textbf{E}}
\renewcommand{\H}{\textbf{H}}

\newtheoremstyle{agdTheorem}{\parskip}{\parskip}{\itshape}{\parindent}{\bfseries}{}{0pt}{\thmname{#1}\thmnumber{~#2}.\thmnote{~\textnormal{#3.}}\quad}
\theoremstyle{agdTheorem}
\newtheorem{theorem}{Theorem}
\newtheorem{lemma}{Lemma}
\newtheorem{proposition}{Proposition}

\newtheoremstyle{agdDefinition}{\parskip}{\parskip}{}{\parindent}{\bfseries}{}{0pt}{\thmname{#1}\thmnumber{~#2}.\thmnote{~\textnormal{#3.}}\quad}
\theoremstyle{agdDefinition}

\newtheorem{remark}{Remark}
\title{Floor Scale Modulo Lifting for QC-LDPC codes}
\author{\IEEEauthorblockN{Nikita Polyanskii, Vasiliy Usatyuk, and  Ilya Vorobyev}
	\IEEEauthorblockA{Huawei Technologies Co.,
		Moscow, Russia\\
		Email: nikitapolyansky@gmail.com,\quad l@lcrypto.com,\quad vorobyev.i.v@yandex.ru}}
\begin{document}
\maketitle
\begin{abstract}

In the given paper we present a novel approach for constructing a QC-LDPC code of smaller length by lifting a given QC-LDPC code. The proposed method can be considered as a generalization of floor lifting. Also we prove several probabilistic statements concerning a theoretical improvement of the method with respect to the number of small cycles. Making some offline calculation of scale parameter it is possible to construct a sequence of QC-LDPC codes with different circulant sizes generated from a single exponent matrix using only floor and scale operations. The only parameter we store in memory is a constant needed for scaling. 
\end{abstract}

\section{Introduction}
Low-density parity-check (LDPC) codes were first discovered by Gallager~\cite{G65}, generalized  by Tanner~\cite{Tanner81}, Wibberg ~\cite{Wibb96}  and rediscovered by MacKay et al.~\cite{McK96,McK99} and Sipser et al.~\cite{SipSpiek96}.

Quasi-cyclic low-density parity-check (QC-LDPC) codes are of great interest to researchers~\cite{KShuFoss01,DjXAGShu03,TanSSFCost04,ChXuDju04,DivDolJ06} since they can be encoded and decoded with low complexity and allow to reach high throughput using linear-feedback shift register~\cite{MYK05,MRM09,ChCo15}.

One advantage of QC-LDPC codes based on circulant permutation matrices (CPM) is that it is easier to analyze their code and graph properties than in the case of random LDPC codes. The performance of LDPC codes is  strongly affected by their graph properties such as the length of the shortest cycle, i.e., girth~\cite{Fossorier04,WDY13}, and 
trapping  sets~\cite{VasCNP09,DDOV15}
and code properties, e.g., the distance of the code~\cite{SmVont09,ButSie13} 
and the ensemble weight enumerator~\cite{Divs06}.


The main contribution of the paper is a novel approach for constructing a quasi-cyclic LDPC code of smaller length by lifting a given QC-LDPC code. The proposed method can be considered as a generalization of floor lifting method introduced in~\cite{myk05,myk06}. 
Making some offline calculation it is possible to construct a sequence of QC-LDPC codes with different circulant sizes generated from a single exponent matrix of QC-LDPC code having the largest length. The only parameter we store in memory is a constant needed for scaling in the lifting procedure.

The outline of the paper is as follows. In Section~\ref{QCLDPC}, we  
introduce some basic definitions  and notations for our presentation. In Section~\ref{LIFTING}, we review state-of-art lifting methods for QC-LDPC codes. Also assuming some natural assumption we prove some probabilistic statements with respect to cycles of length $4$ and provide a comparison between lifting procedures. In Section~\ref{OURLIFT}, we present our floor scale modulo lifting method for QC-LDPC codes and prove several probabilistic statements concerning theoretical improvement of the method with respect to the number of small cycles. The performance of QC-LDPC codes obtained by the floor scale modulo lifting method is investigated by simulations in Section~\ref{SIMULATION}.

\section{QC-LDPC CODES}\label{QCLDPC}
A QC-LDPC code is described by a parity-check matrix $\H$ which consists of square blocks which could be either zero matrix or circulant permutation matrices.  Let $P=(P_{ij})$ be  the $L\times L$ \textit{circulant permutation matrix} defined by
$$
P_{ij}=
\begin{cases}
1,\quad\text{if } i+1\equiv j \mod L\\
0,\quad \text{otherwise}.
\end{cases}
$$
Then   $P^k$ is the circulant permutation matrix (CPM) which shifts the identity matrix $I$ to the right by $i$ times for any $k$, $0\le k\le L-1$. For simplicity of notation denote the zero matrix by $P^{-1}$. Denote the set $\{-1, 0, 1,\ldots, L-1\}$ by $\A_L$. 
Let the matrix $\H$ of size $mL\times nL$ be defined in the following manner
\beq 
 $\H=\left[\begin{array}{cccc} {P^{a_{11} } } & {P^{a_{12} } } & {\cdots } & {P^{a_{1n} } } \\ {P^{a_{21} } } & {P^{a_{22} } } & {\cdots } & {P^{a_{2n} } } \\ {\vdots } & {\vdots } & {\ddots } & {\vdots } \\ {P^{a_{m1} } } & {P^{a_{m2} } } & {\cdots } & {P^{a_{mn} } } \end{array}\right],
\eeq
where $a_{i,j} \in \A_L$. Further we call $L$ the circulant size of $\H$. In what follows a code $C$ with parity-check matrix $\H$ will be referred to as a \textit{QC-LDPC code}.
Let $E(\H) = (E_{ij}(\H))$ be the \textit{exponent matrix} of $\H$ given by:
\beq 
$E(\H)=\left[\begin{array}{cccc} {a_{11} } & {a_{12} } & {\cdots } & {a_{1n} } \\ {a_{21} } & {a_{22} } & {\cdots } & {a_{2n} } \\ {\vdots } & {\vdots } & {\ddots } & {\vdots } \\ {a_{m1} } & {a_{m2} } & {\cdots } & {a_{mn} } \end{array}\right],
\eeq
i.e., the entry $E_{ij}(\H) = a_{ij}$. The \textit{mother matrix} $M(\H)$ is a $m\times n$ binary matrix obtained from replacing $-1$'s and other integers by $0$ and $1$, respectively, in $E(\H)$. If there is a cycle of length $2l$ in the Tanner graph of $M(\H)$, it is called a \textit{block-cycle} of length $2l$.  Any block-cycle in $M(\H)$ of length $2l$ corresponds both to the sequence of $2l$ CPM's $\{P^{a_1},P^{a_2}\ldots, P^{a_{2l}}\}$ in $\H$ and sequence of $2l$ integers $\{a_1, a_2 \ldots a_{2l}\}$ in $E(\H)$ which will be called \textit{exponent chain}. 

The following well known result gives the easy way to find cycles in the Tanner graph of parity-check matrix $H$.
\begin{proposition}[\cite{Fossorier04}]\label{prop::FossCondition}

An exponent chain forms a cycle in the Tanner graph of $H$ iff the following condition holds
$$
\sum_{i=1}^{2l}(-1)^i a_i\equiv 0 \mod L.
$$
\end{proposition}

\section{LIFTING OF QC-LDPC CODES}\label{LIFTING}



\subsection{State-of-art Lifting Methods}
Consider a QC-LDPC code with $mL_0\times nL_0$ parity-check matrix $\H_0$ with circulant size $L_0$, $m\times n$ exponent matrix $E(\H_0) = (E_{ij}(\H_0))$ and mother matrix $M(\H_0)$. Given a set of circulant sizes $\{L_k\}$, $L_k<L_0$, \textit{lifting} is a method of constructing QC-LDPC codes with $mL_k\times nL_k$ parity-check matrices $\H_k$ from $\H_0$, which have the same mother matrix $M(\H_k)=M(\H_0)$ and entries of exponent matrices $E(\H_k)$ satisfy $-1 \le E_{ij}(\H_k)\le L_k -1$. Therefore, it suffices to specify a formula using which we recalculate each value of $E(\H_k)$ from $E(\H_0)$. In paper~\cite{myk05} two lifting approaches are given.

\textbf{Floor lifting} is defined as follows:
\beq{FloorLifting}
E_{ij}(H_{k} )= 
\begin{cases}
\left\lfloor \frac{L_{k} }{L_{0} } \times E_{ij}(H_{0} )\right\rfloor, &\text{if } E_{ij}(H_{0})\neq -1,\\
-1,\quad &\text{otherwise}.
\end{cases}
\eeq

\textbf{Modulo lifting} is determined by the following equation:
\beq{ModuloLifting}
E_{ij}(H_{k} )=\begin{cases}
E_{ij}(H_{0} )\mod L_{k}, &\text{if } E_{ij}(H_{0})\neq -1,\\
-1,\quad &\text{otherwise}.
\end{cases}  
\eeq

Now we prove several probabilistic statements.

Consider an exponent chain of length 4 with exponent values $a$, $b$, $c$, $d$
\begin{center}
$A=\left[\begin{array}{cc} {a} & {b}\\ {c} & {d} \end{array}\right]$,
\end{center}
where each element is chosen independently and equiprobable from the set $\{0, 1, \ldots, 2q - 1\}$, $L_0=2q$ is a circulant size, $q > 2$. Notice that the probability of the event $\C_0$: ``the exponent chain with exponent values $a$, $b$, $c$ and $d$ forms a cycle'', i.e., $a-b-c+d \equiv 0 \mod 2q$,
is equal to $1 / (2q)$. Assume that we use some lifting method to obtain exponent values $a'$, $b'$, $c'$, $d'$
\begin{center}
$B=\left[\begin{array}{cc} {a'} & {b'}\\ {c'} & {d'} \end{array}\right]$,
\end{center}
for circulant size $L_1=q$. We are interested in the probabilities of an event $\C_1$: ``the exponent chain with exponent values $a'$, $b'$, $c'$ and $d'$ forms a cycle'' given the event $\C_0$ and given the event $\overline{\C_0}$. In Sections~\ref{FLOOR} and~\ref{MODULO} we obtain these probabilities for floor lifting and modulo lifting, respectively. Finally, we summarize results and compare these two methods in Section~\ref{FIRSTCONCL}.

\subsection{Floor Lifting} \label{FLOOR}
Let $a = 2a_1 +a_2$, $b = 2b_1 + b_2$, $c = 2c_1 + c_2$ and $d = 2d_1 + d_2$, where $a_2, b_2, c_2, d_2\in \{0,1\}$. One can see that $a'=a_1$, $b'=b_1$, $c'=c_1$, $d'=d_1$. Given the event $\C_0$ occurs, i.e.
$$
2(a_1-b_1-c_1+d_1) + (a_2-b_2-c_2+d_2) \equiv 0 \mod 2q.
$$
the event $\C_1$, i.e., $a_1-b_1+d_1-c_1 \equiv 0 \mod q,$
is equivalent to the condition $a_2-b_2-c_2+d_2=0$.
From $\C_0$ it follows that $a_2-b_2+d_2-c_2  \equiv 0 \mod 2.$
Therefore, the conditional probability
\begin{multline*}
\Pr(\C_1\mid \C_0)=\Pr(a_2-b_2-c_2+d_2 = 0\mid\C_0) =  \\
\Pr(a_2-b_2-c_2+d_2 = 0\mid a_2-b_2-c_2+d_2\equiv 0 \mod 2) = 3/4.
\end{multline*}
Indeed we have exactly $8=2^3$ equiprobable choices for $a_2,b_2,c_2,d_2$ depicted in Table~\ref{choices}, $6=\binom{4}{2}$ of which give the cycle.
\begin{table}[ht]
\caption{Possible choices for $a_2,b_2,c_2,d_2$}
\label{choices}
\begin{center}
\begin{tabular}{|P{15pt}|P{15pt}|P{15pt}|P{15pt}|P{55pt}|}
\hline
$a_2$& $b_2$ & $c_2$ & $d_2$ & $a_2-b_2-c_2+d_2$ \\
\hline 0 &  0 & 0 & 0 & 0 \\
\hline 0 &  0 & 1 & 1 & 0 \\
\hline 0 &  1 & 0 & 1 & 0 \\
\hline 0 &  1 & 1 & 0 & -2 \\
\hline 1 &  0 & 0 & 1 & 2 \\
\hline 1 &  0 & 1 & 0 & 0 \\
\hline 1 &  1 & 0 & 0 & 0 \\
\hline 1 &  1 & 1 & 1 & 0 \\
\hline
\end{tabular}
\end{center}
\end{table}

Now let us find the probability $\Pr(\C_1\mid\overline{\C_0})$. Since 
$$
\Pr(\C_1\mid\overline{\C_0})=\frac{\Pr(\C_1\overline{\C_0})}{\Pr(\overline{\C_0})}
$$ 
and $\Pr(\overline{\C_0})=\frac{2q-1}{2q}$, it suffices to obtain $\Pr(\C_1\overline{\C_0})$. Find the number of all 4-tuples $(a, b, c, d)$, such that $a_1-b_1-c_1+d_1\equiv 0 \mod q$ and 
$a_2-b_2-c_2+d_2\not\equiv 0 \mod 2q$.
We have $q^3$ ways to choose $a_1, b_1, c_1, d_1$ and $10$ ways  to choose $a_2, b_2, c_2, d_2$ for $q>2$. Therefore, 
$$
\Pr(\C_1\overline{\C_0})=\frac{10q^3}{(2q)^4}=\frac{5}{8q}\quad\text{and }\Pr(\C_1\mid\overline{\C_0})=\frac{5}{4(2q-1)}.
$$ 
Let us sum up the results in
\begin{proposition}\label{pr::floorProb}
An exponent chain in $E(\H)$ of length 4, which forms a cycle in the parity-check matrix $\H$ with circulant size $2q$, turns into a cycle in the parity-check matrix $\H'$ with circulant size $q$ obtained after floor lifting with probability $3/4$, while an exponent chain of length 4, which does not form a cycle, turns into a cycle with probability $p_{fl}\eq 5/(4(2q-1))$.
\end{proposition}
\subsection{Modulo Lifting}\label{MODULO}
Let $a = a_1q+a_2$, $b = b_1q+b_2$, $c = c_1q+c_2$, $d = d_1q+d_2$, where $a_2, b_2, c_2, d_2\in\{0,1,\ldots, q-1\}$. It is easy to check that $a'=a_2$, $b'=b_2$, $c'=c_2$, $d'=d_2$. Given the event $\C_0$ occurs, we have 
$$
q(a_1-b_1-c_1+d_1) + (a_2-b_2-c_2+d_2) \equiv 0 \mod 2q.
$$
It follows that 
$$
a'-b'-c'+d'=a_2-b_2-c_2+d_2  \equiv 0 \mod q,
$$
thus the conditional probability $\Pr(\C_1\mid \C_0) = 1.$ 
Let us obtain probability $\Pr(\C_1\mid\overline{\C_0})$. Since 
$$
\Pr(\C_1\mid\overline{\C_0})=\frac{\Pr(\C_1\overline{\C_0})}{\Pr(\overline{\C_0})}
$$ 
and $\Pr(\overline{\C_0})=\frac{2q-1}{2q}$, we need to find $\Pr(\C_1\overline{\C_0})$. Calculate the number of all 4-tuples $(a, b, c, d)$, such that $a_2-b_2-c_2+d_2\equiv 0 \mod q$ and $a-b-c+d\not\equiv 0 \mod 2q$.
We have $q^3$ ways to choose $a_2, b_2, c_2, d_2$ and $8$ ways  to choose $a_2, b_2, c_2, d_2$ for $q>2$. Therefore, 
$$
\Pr(\C_1\overline{\C_0})=\frac{8q^3}{(2q)^4}=\frac{1}{2q}\quad\text{and }
\Pr(\C_1\mid\overline{\C_0})=\frac{1}{2q-1}.
$$

As a result we have obtained the following

\begin{proposition}\label{pr::moduloProb}
An exponent chain in $E(\H)$ of length 4, which forms a cycle in the parity-check matrix $\H$ with circulant size $2q$, turns into a cycle in the parity-check matrix $\H'$ with circulant size $q$ obtained after modulo lifting with probability $1$, while an exponent chain of length 4, which does not form a cycle, turns into a cycle with probability $p_{mod}\eq 1/(2q-1)$.
\end{proposition}

\subsection{Comparison}\label{FIRSTCONCL}
Now summarize the results from Sections~\ref{FLOOR} and~\ref{MODULO} in the following
\begin{theorem} 
Suppose that in exponent matrix $E(\H)$ with circulant size $2q$ we have $y$ exponent chains of length $4$, which do not form a cycle, and $x$ exponent chains of length $4$, which form a cycle. Then mathematical expectations $EC_{fl}$ ($EC_{mod}$) of the number of cycles after floor lifting (modulo lifting) for circulant size $q$ are as follows:
$$
EC_{fl}=\frac{3}{4}x+\frac{5}{4(2q-1)}y,\quad
EC_{mod}=x+\frac{1}{(2q-1)}y.
$$
\end{theorem}

Note that $EC_{fl}\geq EC_{mod}$ when $y\geq (2q-1)x$. Since usually we try to eliminate short cycles in matrix $E(\H)$, the number $y$ is likely to be much greater than $(2q-1)x$. So, we can conclude that modulo lifting is better than floor lifting with respect to  the number of short cycles.

\section{FLOOR SCALE MODULO LIFTING OF QC-LDPC CODES}\label{OURLIFT}
Now we introduce the proposed lifting method which we call \textbf{floor scale modulo lifting}:
\beq{FloorScalemoduloLifting}
E_{ij}(H_{k} )=
\begin{cases}
-1,\quad E_{ij}(H_{0})= -1, \\
\left\lfloor \frac{L_{k} }{L_{0} } ((r \times E_{ij}(H_{0})) \mod L_{0})\right\rfloor, &\text{otherwise,}
\end{cases} 
\eeq
where special parameter $r$ is called a \textit{scale value}. 

Define $A(r)$:
\begin{center}
$A(r)=\left[\begin{array}{cc} {a(r)} & {b(r)}\\ {c(r)} & {d(r)} \end{array}\right]$,
\end{center}
where $$a(r) \equiv r a \mod 2q, \quad b(r) \equiv r b \mod 2q,
$$ 
$$
c(r) \equiv r c \mod 2q, \quad d(r) \equiv r d \mod 2q.
$$
By $\C_0(r)$ denote the event:
``the exponent chain with exponent values $a(r)$, $b(r)$, $c(r)$ and $d(r)$ forms a cycle''.
Notice that for $r$ coprime with $2q$, i.e. $(r, 2q)=1$, elements of matrix $A(r)$ have the same distribution as matrix $A$. Moreover, exponent chains from matrices $A$ and $A(r)$ form a cycle simultaneously.
Let $a = 2a_1 +a_2$, $b = 2b_1 + b_2$, $c = 2c_1 + c_2$ and $d = 2d_1 + d_2$, where $a_2, b_2, c_2, d_2\in \{0,1\}$. Suppose we use floor scale modulo lifting for $L_1 = q$ with scale value $r=2t+1$, $0<r<2q$, which is coprime with $2q$. Then we obtain matrix $B(r)$:
\begin{center}
$B(r)=\left[\begin{array}{cc} {a'(r)} & {b'(r)}\\ {c'(r)} & {d'(r)} \end{array}\right]$,
\end{center}
where 
$$
a'(r)=\left\lfloor \frac{2a_1 r+a_2 r}{2} \right\rfloor \equiv a_1r + a_2 t \mod q.
$$
Other values $b'(r)$, $c'(r)$ and $d'(r)$ are represented in the same way. By $\C_1(r)$ denote the event:
``the exponent chain with exponent values $a'(r)$, $b'(r)$, $c'(r)$ and $d'(r)$ forms a cycle''. One can see that
$$
\Pr(\C_1(r)\mid \C_0) =\Pr(\C_1(1)\mid \C_0) =  \frac{3}{4}.
$$
Moreover 
$$
\C_1(r)\cap \C_0 = \C_1(1)\cap \C_0.
$$
\begin{proposition}\label{pr::twoRs}
Let $r_1$, $r_2$ be two distinct integers, such that $0<r_1, r_2<2q$, $(r_1, 2q)=1$, $(r_2, 2q)=1$ and $r_1\not\equiv r_2(q+1) \mod 2q$.
Then
$$
\Pr(\C_1(r_1)\cap \C_1(r_2) \cap \overline{\C_0})=0.
$$
\end{proposition}
In other words, for any scale values $r_1$ and $r_2$ fulfilled the condition of Proposition~\ref{pr::twoRs} if the start exponent chain in the matrix $A$ does not form a cycle then at least one exponent chain in the matrices $B(r_1)$ and $B(r_2)$ does not form a cycle too.
\begin{proof}
Let $u_1$ be such integer that $u_1r_1\equiv 1 \mod 2q$.  Note that $\C_0(u_1)=\C_0$. Therefore, 
$$
\Pr(\C_1(r_1)\cap \C_1(r_2) \cap \overline{\C_0})=\Pr(\C_1(r_1u_1)\cap \C_1(r_2u_1) \cap \overline{\C_0})=0.
$$
Assume events $\C_1(r_1u_1)$ and $\C_1(r_2u_1)$ occur. Thus, $a_1-b_1-c_1+d_1\equiv 0 \mod q$
and 
$$
(a_1-b_1-c_1+d_1)r'_2+(a_2-b_2-c_2+d_2)t'_2\equiv 0 \mod q,
$$
where 
$$
1+2t'_2=r'_2\equiv r_2u_1 \mod q, \quad 0<r'_2<2q.
$$
From  
$$
(a_2-b_2-c_2+d_2)t'_2\equiv 0 \mod q, 
$$
$$
(a_2-b_2-c_2+d_2)\in[-2,2], \quad t'_2\in[1, q-1]
$$ and $2t'_2+1\ne q+1$ it follows that $a_2-b_2-c_2+d_2=0$. Hence $(a_1-b_1-c_1+d_1)\equiv 0\mod q$ 
and 
\begin{multline*}
2(a_1-b_1-c_1+d_1)+(a_2-b_2-c_2+d_2) \\
=a-b-c+d\equiv 0 \mod 2q,
\end{multline*}
i.e., we prove that $\C_1(r_1)\cap \C_1(r_2)\Rightarrow C_0$.
\end{proof}

\begin{remark}
Note that if $r_1\equiv r_2(q+1)\mod 2q$, then $r_2\equiv r_1(q+1)\mod 2q$. 
Therefore, we can choose a set $R$ of scale values of cardinality $\varphi(2q)/2$ ($\varphi(n)$ is Euler's totient function)  for even $q$ and $\varphi(2q)$ for odd $q$, such that for every $r_1, r_2\in R$ the conditions of Proposition~\ref{pr::twoRs} are fulfilled.
\end{remark}

Consider a floor scale modulo lifting with a family $R=\{r_1,r_2,\ldots, r_{N_r}\}$ of $N_r$ scale values, such that for any two scale values $r_i,\, r_j\in R$ the conditions of Proposition~\ref{pr::twoRs} are satisfied. Let $D=(D_{ij})$ be an $N_r\times y$ matrix, where the $i$-th row corresponds to scale values $r_i\in R$, and each column corresponds to one exponent chain  of length $4$ in $E(\H)$. We set $D_{ij}$ to $1$ if the $j$-th exponent chain forms a cycle  after floor scale modulo lifting with scale value $r_i$, and to $0$ otherwise. The first $x$ columns, which corresponds to cycles in exponent matrix with circulant size $2q$, equal to the column of ones with probability $3/4$ and to the column of zeros with probability $1/4$. The rest $y$ columns equal to the column of zeros with probability 
$$
1-N_r p_{fl} = 1-\frac{5N_r}{4(2q-1)}
$$
and to the column of weight 1 with one at position $i$ with probability $p_{fl} =5/(4(2q-1))$
for each $i\in [1, N_r]$. Let $X_i$ be equal to the number of ones in the $i$-th row. We are interested in the minimum number of cycles $\min(X_1, X_2, \ldots, X_{N_r})$. For further calculations
we assume that all columns of matrix $D$ are chosen independently. Under this assumption exact formulas for the mathematical expectation $EC_{fsml}(N_r)= 3x/4  + \E\min(X_1, X_2, \ldots, X_{N_r})$ could be easily written out, but they rather messy. We provide only formula for the case $N_r=2$ in the form of 
\begin{proposition}[\cite{puv17}]\label{pr::floorProb}
Suppose we have an exponent matrix $E(\H)$ with circulant size $2q$ having $x$ exponent chains of length $4$, which form a cycle in $\H$, and $y$ exponent chains of length $4$, which do not form a cycle. Then the mathematical expectation  $EC_{fsml}(2)$ of the number of cycles of length $4$ in the parity-check matrix of circulant size $q$ obtained after floor scale modulo lifting with $N_r=2$ scale values, which satisfies the conditions of Proposition~\ref{pr::twoRs}, is described by the following expression
\begin{multline*}
EC_{fsml}(2)=\frac{3}{4}x+\sum\limits_{n=0}^{y}\frac{n}{2}\(1-\frac{\binom{n}{\lfloor\frac{n}{2}\rfloor}}{2^n}\)\\
\times(2p_{fl})^n(1-2p_{fl})^{y-n}\binom{y}{n}.
\end{multline*}
\end{proposition}

The proof of Proposition~\ref{pr::floorProb} is provided in the full version of the given paper~\cite{puv17}.

If $y\to\infty$ the asymptotic behavior of $EC_{fsml}(N_r)$ is given by
\begin{theorem}[\cite{puv17}]\label{th::asymManyCycles}
The mathematical expectation of the number of cycles of length 4 after floor scale modulo lifting has the  following asymptotic form
$$
EC_{fsml}=\frac{3}{4}x+p_{fl}y-c_{N_r}\sqrt{y}+o(\sqrt{y}), \text{ if }y\to\infty, 
$$
where $c_{N_r}$ does not depend on $y$. 
\end{theorem}

Let us consider another scenario. Suppose that the number of cycles of length 4 in matrix $\H$ with circulant size $L_0=2q$ is equal to $0$, and the number $y$ of exponent chains is fixed. Now we are interested in the probability that after lifting for the circulant size $L_1=q$ we will not obtain any cycle of length 4. We again assume that all events $C_1$ are independent for all exponent chains, i.e., all columns of matrix $D$ are chosen independently.

\begin{theorem}[\cite{puv17}]\label{th::asymBigCirc}
The probability of the absence of cycles of length $4$ in the parity-check matrix with circulant size $q$ obtained after modulo lifting, floor lifting and floor scale modulo lifting is as follows
$$
P_{mod} = (1-p_{mod})^y=1-yp_{mod}+O(q^{-2}), \quad q\to\infty
$$
$$
P_{fl} = (1-p_{fl})^y=1-yp_{fl}+O(q^{-2}), \quad q\to\infty
$$
\begin{multline*}
P_{fsml}(N_r)=\sum\limits_{k=1}^{N_r}(-1)^{k-1}\binom{N_r}{k}(1-kp_{fl})^y\\
=\begin{cases}
1-O(q^{-N_r}), \quad\text{if }y\ge N_r,\,q\to\infty,
\\
1, \quad\text{if }y< N_r,\,q\to\infty.
\end{cases}
\end{multline*}
\end{theorem}

In this case we see that floor scale modulo lifting is much better than modulo and floor lifting.

Table~\ref{BrokenCycles} shows one of possible advantages of the
proposed lifting approach.  We compare the floor lifting length adaption of QC-LDPC codes used in IEEE $802.16$ for  rate $1/2$ with the proposed floor scale modulo lifting. We apply the lifting methods to the $12 \times 24$ mother matrix. We have found optimal $r$ scale value for our lifting approach with respect to girth and number of exponent chains which form cycles of the minimal length. In Table~\ref{BrokenCycles} for each circulant size the optimal $r$ scale value, girth and the number of cycles are depicted. Note that the QC-LDPC code of  IEEE $802.16$ standard was optimized with considering floor lifting method. If the QC-LDPC code with the maximal length size is not optimized with considering floor or modulo lifting method, then the superiority of the proposed floor scale modulo lifting will be more conspicuous.
\begin{table}[ht]
\caption{Girth and the number of short cycles for floor lifting and floor scale modulo lifting}
\label{BrokenCycles}
\begin{center}
\begin{tabular}{|P{10pt}|P{10pt}|P{45pt}|P{10pt}|P{45pt}|}
\hline
\multicolumn{3}{|P{90pt}|}{Floor scale modulo lifting}   &
\multicolumn{2}{|P{70pt}|}{Floor lifting, $r=1$} \\
\hline
$E_{k}$& $r$ & girth / cycles& $E_{k}$& girth / cycles\\
\hline 24&  95& 6 / 13 & 24& 6 / 20 \\
\hline 28& 1& 4 / 1 & 28& 4 / 1 \\
\hline 32& 1& 6 / 11 & 32& 6 / 11 \\
\hline 36& 95& 6 / 7 & 36& 6 / 13 \\
\hline 40& 1& 6 / 7 & 40& 6 / 7 \\
\hline 44& 95& 6 / 5 & 44& 6 / 10 \\
\hline 48& 1& 6 / 7 & 48& 6 / 7 \\
\hline 52& 1& 6 / 6 & 52& 6 / 6 \\
\hline 56& 1& 6 / 5 & 56& 6 / 5 \\
\hline 60& 1& 6 / 6 & 60& 6 / 6 \\
\hline 64& 34& 6 / 5 & 64& 6 / 9 \\
\hline 68& 53& 6 / 4 & 68& 6 / 8 \\
\hline 72& 11& 6 / 6 & 72& 6 / 9 \\
\hline 76& 91& 6 / 4 & 76& 6 / 5 \\
\hline 80& 2& 6 / 5 & 80& 6 / 7 \\
\hline 84& 11& 6 / 3 & 84& 6 / 8 \\
\hline 88& 41& 6 / 3 & 88& 6 / 6 \\
\hline 92& 13&  6 / 4 &  92&  6 / 8 \\
\hline 96&  1&  6 / 5 &  96&  6 / 5 \\
\hline
\end{tabular}
\label{Number of minimal cycles in QC-LDPC codes lifted by floor scale and floor scale modulo lifting}
\end{center}
\end{table}
\section{SIMULATION RESULTS}\label{SIMULATION}
QC-LDPC codes of smaller lengths can be obtained by lifting exponent matrix of  QC-LDPC codes of maximal length. Their performance over an AWGN channel with BPSK modulation was analyzed by computer simulations. 
Figure~\ref{Fig1} shows the frame error rate (FER) performance of rate 4 over 5 AR4JA code defined by protograph of size $3 \times 11 $  from~\cite{DDJ06}. We use native lifting for fixed circulant sizes $\{16, 32, 64, 128\}$ and floor modulo scale lifting  beginning from parity-check matrix $\H$ of circulant size $128$ which goes down to circulant sizes $\{16, 32, 64\}$. BP decoder with $100$ iterations is used.

Figure~\ref{Fig2} shows the SNR required to achieve $10^{-2}$ FER performance over an AWGN channel with QPSK modulation
for 3 families of QC-LDPC codes with rate 8 over 9.
Families $A$~\cite{SAMd16} and $B$~\cite{QCM16} are the industrial state-of-art QC-LDPC codes with their own lifting. For family $C$ we applied floor modular scale lifting.
Layered normalized offset min-sum decoder with 15 iterations was used in simulations. Normalize and offset factors were optimized  to improve waterfall performance~\cite{CTanJL05}. In summary, the proposed lifting scheme supports fine granularity and avoids catastrophic cases for different lengths.
\begin{figure}
\centering
\includegraphics[width=83mm, viewport=30.00mm 198.12mm 194.93mm 277.00mm]{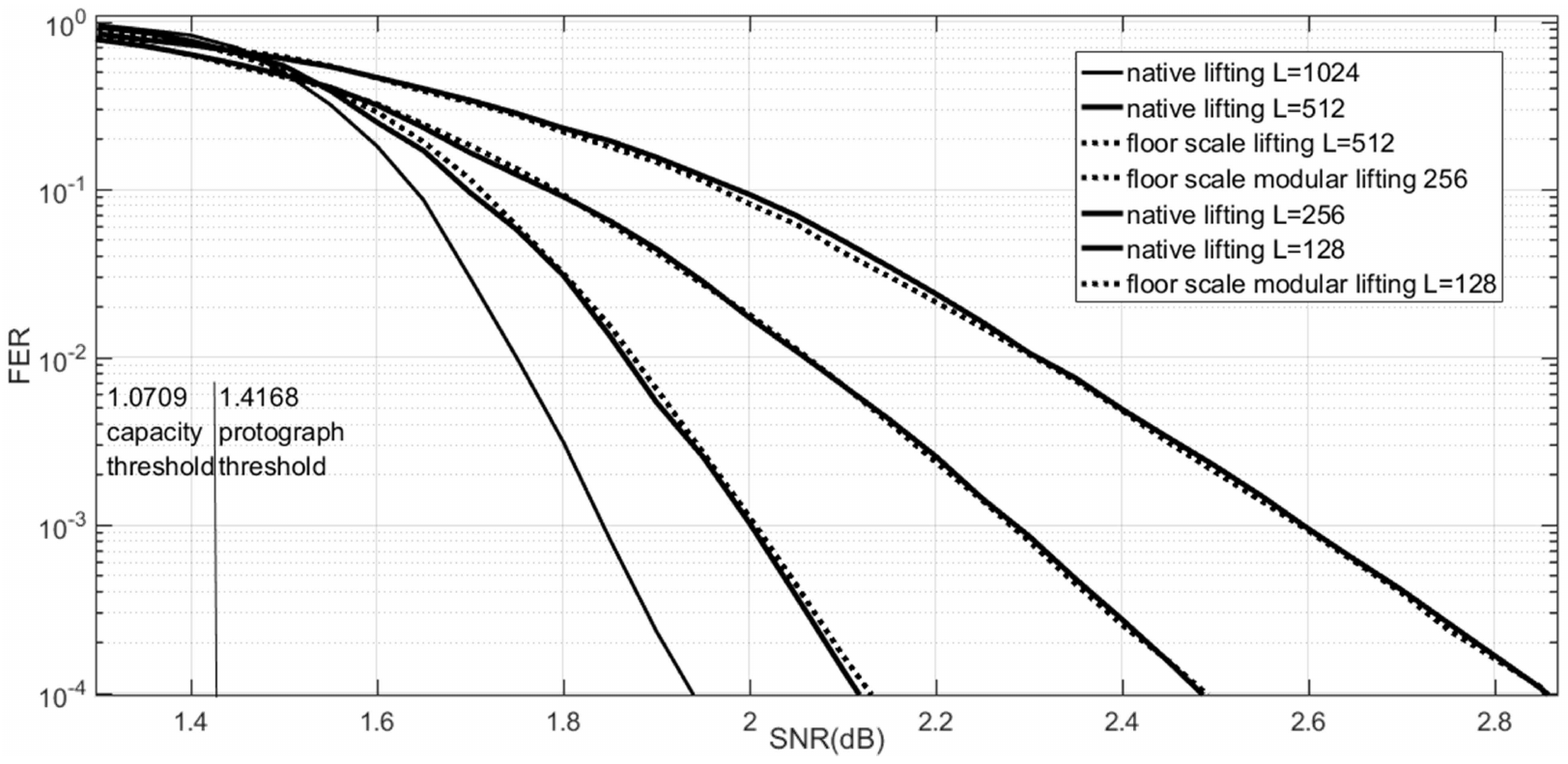} 
  \caption{\textsf{Performance of AR4JA codes with different CPM sizes by the lifting L=128, 256, 512. BP 100 it.}}
  \label{Fig1}
\end{figure}

\begin{figure}
\centering
\includegraphics[width=83mm, viewport=30.00mm 197.40mm 194.75mm 277.00mm]{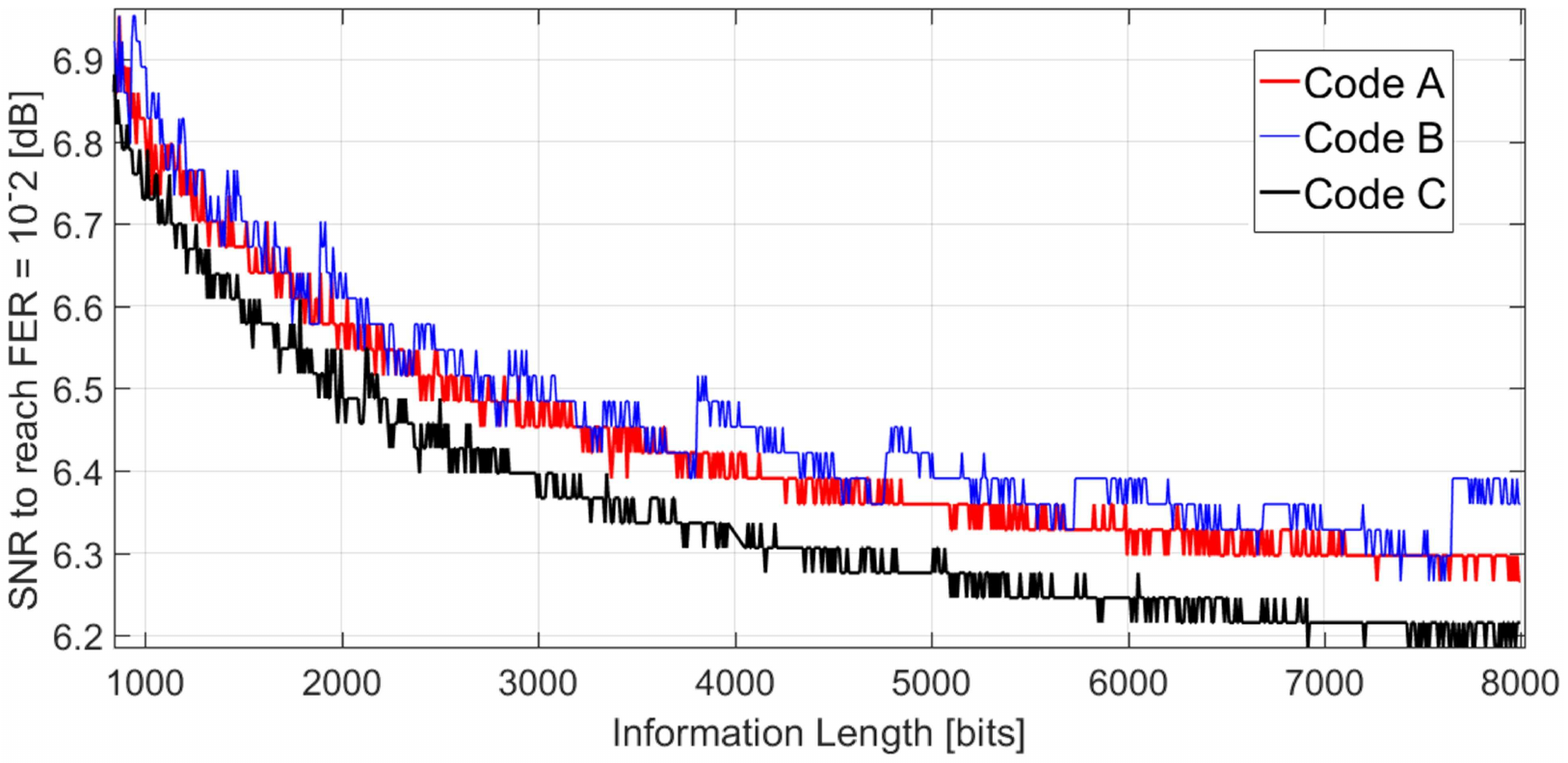} 
  \caption{\textsf{SNR required to achieve $10^{-2}$ FER}}
  \label{Fig2}
\end{figure}

\appendix\label{appendix}
\subsection{Proof of Proposition~\ref{pr::floorProb}}
\begin{proof}
Consider a random variable $\min(X_1,X_2)$. Then
\begin{multline*}
\E \min(X_1,X_2) = \E\,\E(\min(X_1,X_2)\mid X_1 + X_2)\\
=\sum_{n=0}^{y}\E\min(Y,n-Y)\Pr(X_1 + X_2 = n),
\end{multline*}
where $Y\sim B(n,\frac{1}{2})$. From
$$
\E\min(Y,n-Y) = \frac{n}{2}\(1-\frac{\binom{n}{\lfloor n/2 \rfloor}}{2^n}\)
$$
and
$$
\Pr(X_1 + X_2 = n) =(2p_{fl})^n(1-2p_{fl})^{y-n}\binom{y}{n},
$$
it follows the statement of the proposition.
\end{proof}
\subsection{Proof of Theorem~\ref{th::asymManyCycles}}
\begin{proof}
Consider a random vector 
$$
{\bf X}=(X_1, X_2, \ldots, X_{N_r}, X_{N_r+1}),
$$
where $X_{N_r+1} = y-\sum\limits_{i=1}^{N_r}X_i)$.

One can see that ${\bf X}$ has a multinomial distribution with 
$$
\Pr(x_1, \ldots, x_{N_r+1})=\frac{y!}{x_1!\cdot\ldots\cdot x_{N_r+1}!}p_{fl}^{\sum\limits_{i=0}^{N_r}x_i}(1-N_rp_{fl})^{x_{N_r+1}}.
$$
By the Central limit theorem the distribution of random vector $\frac{{\bf X}-\E {\bf X}}{\sqrt{y}}$ tends to normal distribution $\mathcal{N}(0, \Sigma)$ as $y\to\infty$, where $\Sigma$ is the covariance matrix of ${\bf X}$.

Let us prove that $\E \min(X_1, X_2, \ldots, X_{N_{r}}) - \E \min {\bf X} = o(\sqrt{y})$ as $y \to\infty$. 
\begin{multline}
|\E \min(X_1, X_2, \ldots, X_{N_{r}}) - \E \min {\bf X}|\\
\leq y\P(\min {\bf X} = X_{N_r+1})\leq y\P(X_{N_r+1}\leq y/({N_{r}+1}))\\
\leq
y\P\left(X_{N_r+1}\leq \E X_{N_r+1}\frac{1}{(1-p_{fl}N_r)(N_r+1)}\right)
\end{multline}

From $N_r\leq \varphi(2q)\leq q$ it follows that $1-p_{fl}N_r>p_{fl}$ for $q>2$. The following chain of inequalities takes place
\begin{multline}
1-p_{fl}N_r>p_{fl} \Rightarrow \\
p_{fl}(N_r+1)N_r\leq N_r<N_r+1\Rightarrow \\
\frac{1}{(1-p_{fl}N_r)(N_r+1)}<1.
\end{multline}
Denote $\frac{1}{(1-p_{fl}N_r)(N_r+1)}$ as $1-\delta$ for some $\delta>0$ and use the Chernoff inequality
$$
y\P\left(X_{N_r+1}\leq \E X_{N_r+1}(1-\delta)\right)\leq ye^{-\frac{\delta^2\E X_{N_r+1}}{2}}=o(\sqrt{y}).
$$

Let denote as  $\min_n({\bf X})$ the function which equals $\max(\min(\min({\bf X}), n), -n)$. This function is continuous and bounded, hence
\begin{multline}
\lim_{y\to\infty}\E\min\left(\frac{{\bf X}-\E {\bf X}}{\sqrt{y}}\right)=\\
\lim\limits_{n\to\infty}\lim_{y\to\infty}\E\min{}_n\frac{{\bf X}-\E {\bf X}}{\sqrt{y}}\\
=\lim\limits_{n\to\infty}\E\min{}_n\mathcal{N}(0, \Sigma)=\\
\E\min\mathcal{N}(0, \Sigma)=-c_{N_r}, 
\end{multline}
where $c_{N_r}=\E\max\mathcal{N}(0, \Sigma)>0.$
Therefore, 
\begin{multline}
\E \min(X_1, X_2, \ldots, X_{N_r}) = \E\min{\bf X} + o(\sqrt{y}) =\\
\min\E{\bf X}-\sqrt{y}c_{N_r}+o(\sqrt{y})=p_{fl}y-\sqrt{y}c_{N_r}+o(\sqrt{y}).
\end{multline}
\end{proof}

\subsection{Proof of Theorem~\ref{th::asymBigCirc}}
\begin{proof}
We prove the formula for $P_{fsml}$ only.  Firstly, find the probability $P_k$ that the first $k$ rows of the matrix $D$ contain only zeros
$$
P_k=(1-kp_{fl})^y.
$$
Secondly, using the inclusion-exclusion principle we get
\beq{pfsml}
P_{fsml} = \sum_{k=1}^{N_r}(-1)^{k-1}\binom{N_r}{k}(1-kp_{fl})^y.
\eeq
Now we prove an auxiliary statement.
\begin{lemma}\label{auxLemma}
The binomial identity
$$
\sum_{k=0}^n(-1)^k\binom{n}{k}g(k)=0
$$
holds for every polynomial $g(k)$ with degree less than $n$.
\end{lemma}
\begin{proof}
Every polynomial $g(k)$ of degree $t$ can be represented in the following form
$$
g(k) = \sum_{l=0}^tc_l(k)_l,
$$
where
$$
(k)_l=k(k-1)\ldots(k-l+1),
$$
ans $c_l$ are some coefficients.

For every $l<n$
\begin{multline}
\sum_{k=0}^n(-1)^k\binom{n}{k}(k)_l=
\sum_{k=l}^n(-1)^k\binom{n}{k}(k)_l\\
=
\sum_{k=l}^n(-1)^k\frac{n!k!}{k!(n-k)!(k-l)!}=
\sum_{k=l}^n(-1)^k\frac{n!k!}{k!(n-k)!(k-l)!}
\\
=
\frac{n!}{(n-l)!}\sum_{k=l}^n(-1)^k\frac{(n-l)!}{(n-k)!(k-l)!}\\
=\frac{n!}{(n-l)!}\sum_{k=l}^n(-1)^k\binom{n-l}{n-k}=0,
\end{multline}
therefore, 
$$
\sum_{k=0}^n(-1)^k\binom{n}{k}g(k)=0.
$$
\end{proof}
Finally, using the evident asymptotic
$$
p_{fl} =\frac{5}{4(2q-1)}=O\(\frac{1}{q}\),\quad \text{as } q\to\infty,
$$
along with Lemma~\ref{auxLemma} and the equality~\req{pfsml}, we obtain the statement of Theorem~\ref{th::asymBigCirc}.
\end{proof}

\end{document}